\documentclass[a4paper,USenglish,thm-restate,final]{lipics-v2019}

\usepackage{todonotes}
\presetkeys{todonotes}{inline}{}

\newcommand{\EXPSPACE}{\textsc{ExpSpace}{}}
\newcommand{\NLOGSPACE}{\textsc{NLogSpace}{}}
\newcommand{\NP}{\textsc{NP}}

\newcommand{\N}{\mathbb{N}}
\newcommand{\Z}{\mathbb{Z}}
\newcommand{\R}{\mathbb{R}}
\newcommand{\Q}{\mathbb{Q}}

\newcommand{\tuple}[1]{\langle #1 \rangle}
\newcommand{\set}[1]{\{ #1 \}}

\newcommand{\VASS}{\textsc{VASS}}

\newcommand{\allConditions}{\textbf{(U1)}, \textbf{(B1)} and \textbf{(BU)}}
\newcommand{\limavginfS}[1]{\textsc{LimAvgInf}_{#1}}
\newcommand{\limavgsupS}[1]{\textsc{LimAvgSup}_{#1}}

\newcommand{\expectedGain}{\mathbb{E}(Gain)}
\newcommand{\Bounded}{\mathbf{B}}
\newcommand{\UnBounded}{\mathbf{U}}

\newcommand{\vass}{\mathcal{A}}

\newcommand{\comp}{\pi}

\newcommand{\keyIdeas}{\subparagraph*{Key ideas.}}

\newcommand{\connectA}{\textbf{in}}
\newcommand{\connectB}{\textbf{out}}
\newcommand{\cycleI}{\mathbf{c}}
\newcommand{\cycleD}{\mathbf{d}}
\newcommand{\cycleE}{\mathbf{e}}
\newcommand{\cycleF}{\mathbf{f}}

\nolinenumbers
\hideLIPIcs

\title{Multi-dimensional Long-Run Average Problems for Vector Addition Systems with States}
\titlerunning{Multi-dimensional Long-Run Average Problems for VASS}

\author{Krishnendu Chatterjee}{IST Austria, Austria}{krish.chat@ist.ac.at}{https://orcid.org/0000-0002-4561-241X}{The Austrian Science Fund (FWF) NFN grant S11407-N23 (RiSE/SHiNE).} 
\author{Thomas A. Henzinger}{IST Austria, Austria}{tah@ist.ac.at}{}{The Austrian Science Fund (FWF) grants S11402-N23 (RiSE/ShiNE) and Z211-N23 (Wittgenstein Award).}
\author{Jan Otop}{University of Wroc\l{}aw, Poland}{jotop@cs.uni.wroc.pl}{https://orcid.org/0000-0002-8804-8011}{The National Science Centre (NCN), Poland under grant 2017/27/B/ST6/00299.}

\authorrunning{K.\ Chatterjee, T.A.\ Henzinger, J.\ Otop}

\Copyright{Krishnendu Chatterjee, Thomas A. Henzinger and Jan Otop}


\acknowledgements{We want to thank the anonymous reviewers of CONCUR 2020 for their helpful comments, which
contributed to the final version of this paper.}


\newcommand{\Path}{\textbf{p}}
\newcommand{\finPath}{\textbf{s}}
\newcommand{\gain}{\textsc{Gain}}
\newcommand{\prob}{\mathbb{P}}
\newcommand{\cost}{\textbf{c}}
\newcommand{\expected}{\mathbb{E}}
\newcommand{\compFin}{\rho}
\newcommand{\thresholds}{\vec{\lambda}}
\newcommand{\markov}{\mathcal{M}}
\newcommand{\avgS}[1]{\textsc{Avg}_{#1}}
\newcommand{\limavgS}[1]{\textsc{LimAvg}_{#1}}
\newcommand{\randComp}{\xi}

\begin{document}

\maketitle

\begin{abstract}
A vector addition system with states (VASS) consists of a finite set of states and counters. 
A transition changes the current state to the next state, and every counter is 
either incremented, or decremented, or left unchanged.
A state and value for each counter is a configuration; and
a computation is an infinite sequence of configurations with transitions between successive configurations. 
A probabilistic VASS consists of a VASS along with a probability distribution over
the transitions for each state.
Qualitative properties such as state and configuration reachability have been widely 
studied for VASS.
In this work we consider multi-dimensional long-run average objectives for VASS and 
probabilistic VASS. 
For a counter, the cost of a configuration is the value of the counter;
and the long-run average value of a computation for the counter is the long-run average 
of the costs of the configurations in the computation.
The multi-dimensional long-run average problem given a VASS and a threshold value for each counter,
asks whether there is a computation such that for each counter the long-run average 
value for the counter does not exceed the respective threshold.
For probabilistic VASS, instead of the existence of a computation, we consider whether 
the expected long-run average value for each counter does not exceed the respective threshold.
Our main results are as follows: we show that the multi-dimensional long-run average problem 
(a)~is NP-complete for integer-valued VASS;  
(b)~is undecidable for natural-valued VASS (i.e., nonnegative counters); and 
(c)~can be solved in polynomial time for probabilistic integer-valued VASS, and 
probabilistic natural-valued VASS when all computations are non-terminating. 
\end{abstract}

\section{Introduction}\label{sec-intro}

\subparagraph*{Vector Addition System with States (VASS) and probabilistic VASS.}
{\em Vector Addition Systems (VASs)} provide a powerful framework for analysis of 
parallel processes~\cite{EN94}. 
They are equivalent to the well-studied model of Petri Nets~\cite{KM69}.
The generalization of VASs with a finite-state transition system gives {\em Vector Addition Systems with States (VASS)}.
The model of VASS is as follows: there is a finite set of control states with 
transitions between them, and a set of $k$ counters, where at every transition between 
the control states each counter is either incremented, decremented, or remains unchanged.
For a VASS, a {\em configuration} is a control state and a valuation of each counter, 
and the transitions of the VASS determines the transitions between the configurations.
Thus a VASS is a finite description of an infinite-state transition system between 
the configurations.
The class of VASS where the counters can hold all possible integer values, are referred to as 
integer-valued VASS; and the class of VASS where the counters can hold only non-negative values,  
are referred to as natural-valued VASS.
A probabilistic VASS consists of a VASS along with probability distribution over the transitions 
for every state.

\subparagraph*{VASS Framework in Verification.}
VASS are an elegant mathematical framework for concurrent processes~\cite{EN94}, and 
have been widely studied in performance analysis of concurrent 
processes~\cite{DKO13:conc-verification-vass,GM12:asynchronous-verification-TOPLAS,KKW10:dynamic-cutoff-detection,KKW12:coverability-proof-minim}.
They have also been used in several other contexts, such as: 
(a)~analysis of parametrized systems~\cite{Bloem16},
(b)~abstract models for programs for bounds  analysis~\cite{SZV14}, 
(c)~interactions between components of an API in component-based 
synthesis~\cite{FMWDR17:component-based-synthesis}.
The~probabilistic VASS provide a natural model for problems mentioned above 
with stochasticity in the system~\cite{DBLP:conf/lics/BrazdilKKN15}.
Thus VASS and probabilistic VASS provide a rich framework for many 
problems in verification and program analysis.

\subparagraph*{Previous results for VASS.}
A computation (run) in a VASS is an infinite sequence of configurations with transitions between
successive configurations. 
The classical problems studied for VASS are as follows:
(a)~{\em control-state reachability} where given a set of target control states a computation satisfies the 
objective if a target state is reached; 
(b)~{\em configuration reachability} where given a set of target configurations a computation satisfies the objective
if a target configuration reached. 
For natural-valued VASS,
(a)~the control-state reachability problem is \EXPSPACE-complete: the 
\EXPSPACE-hardness is shown in~\cite{Esparza:PN,lipton1976reachability} 
and the upper bound follows from~\cite{RACKOFF1978223}; and 
(b)~the configuration reachability problem is decidable~\cite{Kosaraju82,DBLP:journals/tcs/Lambert92,leroux2012vector,DBLP:conf/stoc/Mayr81}, 
and a recent breakthrough result establishes non-elementary hardness~\cite{DBLP:conf/stoc/CzerwinskiLLLM19}.
For integer-valued VASS,
(a)~the control-state reachability problem is \NLOGSPACE-complete (by reduction to graph reachability);
(b)~the configuration reachability problem is \NP-complete. 
In probabilistic VASS, for the natural-valued class, even defining the probability measure 
over infinite computations is a challenging and complex problem~\cite{DBLP:conf/lics/BrazdilKKN15}, as computations 
that violate the non-negativity condition terminate as finite computations.

\subparagraph*{Long-run average objective and multi-dimensional long-run average problem.}
The classical problems for VASS consider qualitative (or Boolean) objectives where each computation is
either satisfactory or not.
In this work we consider multi-dimensional long-run average objective. 
For a counter, we consider the cost of a configuration as the value of the counter.
For a computation, the long-run average of the costs of the configurations of the computation is the 
long-run average value for the respective counter.
The multi-dimensional long-run average problem given a VASS and a threshold value for each counter,
asks whether there is a computation such that for each counter the long-run average 
value for the counter does not exceed the respective threshold.
For integer-valued probabilistic VASS, instead of the existence of a computation, we consider whether 
the expected long-run average value for each counter does not exceed the respective threshold.
For natural-valued probabilistic VASS, the presence of terminating runs makes even defining the 
probability measure complex.
We consider two variants: (a)~\emph{strict semantics} that require all computations to be non-terminating;
(b)~\emph{relaxed semantics} where we consider the conditional probability with respect to 
non-terminating runs.

\subparagraph*{Motivating examples.}
We present some motivating examples for the problems we consider.
First, consider a VASS where the counters represent different queue lengths, and each queue consumes 
a resource type (e.g., energy or memory or time delay) proportional to its length.
The multi-dimensional long-run average problems asks that the average consumption of each resource
does not exceed a desired threshold.
Second, consider a system that uses two different batteries, and the counters represent the charge levels.
At different states, different batteries are used, and we are interested in the long-run average 
charge of each battery. 
This is again modeled as the multi-dimensional long-run average problem.

\subparagraph*{Our contributions.} 
Our main contributions are as follows:
\begin{enumerate}

\item For non-probabilistic VASS we show that the multi-dimensional long-run average problem 
(a)~is \NP-complete for integer-valued VASS, and (b)~is undecidable for natural-valued VASS.

\item For probabilistic integer-valued VASS, we show that the 
multi-dimensional long-run average problem can be solved in polynomial time.
For natural-valued VASS, we show that the multi-dimensional problem can be solved
in polynomial-time for (a)~the strict semantics, and (b)~the relaxed semantics for strongly 
connected VASS such that the expected multi-dimensional long-run average is finite.
For the relaxed semantics and general natural-valued VASS, we show \EXPSPACE-hardness,
and the exact decidability and complexity remain open.
\end{enumerate}

\subparagraph*{Related works.}
For probabilistic VASS the long-run average behavior problem has been studied~\cite{DBLP:conf/lics/BrazdilKKN15},
as well as for other infinite-state models such as pushdown automata and games~\cite{AAHMKT14,CV17a,CV17b}.
However, these works consider that costs are associated with the transitions of the finite-state 
system and do not depend on the counter values; moreover, they do not consider the multi-dimensional problem.
In contrast, we consider costs that depend on the counter values, and hence on the configurations.
Costs based on configurations, specifically the content of the stack in pushdown automata, have been considered in~\cite{DBLP:conf/fsttcs/MichaliszynO17}.
Quantitative asymptotic bounds for polynomial-time termination in VASS have also been studied~\cite{BCKNVZ18,Ler18},
however, these works do not consider long-run average property.
Finally, a related model of automata with monitor counters with long-run average property have been considered
in~\cite{CHO16b,CHO16a}.
However, there is a crucial difference: in automata with monitor counters, counters are reset once the value is used. 
Moreover, the complexity results for automata with monitor counters are quite different from the results
we establish. 
Finally a recent work considers long-run average problem for VASS~\cite{DBLP:conf/concur/ChatterjeeHO19}. 
However, the cost is always single-dimensional with a linear combination of the counter values, and moreover,
probabilistic VASS have not been considered in~\cite{DBLP:conf/concur/ChatterjeeHO19}.

\section{Preliminaries}
For a sequence $w$, we define $w[i]$ as the $(i+1)$-th element of $w$ (we start with $0$)
and $w[i,j]$ as the subsequence $w[i] w[i+1] \ldots w[j]$. 
We allow $j$ to be $\infty$ for infinite sequences.
For a finite sequence $w$, we denote by $|w|$ its length; and for an
infinite sequence the length is $\infty$.
We use the same notation for vectors. 
For a vector $\vec{x} \in \R^k$  (resp., $\Q^k$, $\Z^k$ or $\N^k$), we define $x[i]$ as the $i$-th component of $\vec{x}$.

\subsection{Vector addition systems with states (VASS)}

A $k$-dimensional vector addition system with states (VASS) over $\Z$ (resp., over $\N$), referred to as $\VASS(\Z,k)$ (resp., $\VASS(\N,k)$), is 
a tuple $\vass = \tuple{Q, Q_0, \delta}$, where
(1)~$Q$ is a finite set of states,
(2)~$Q_0 \subseteq Q$ is a set of initial states, and
(3)~$\delta \subseteq Q \times Q \times \Z^k$ is a transition relation.
In a transition $(q,q',\vec{y})$, the vector $\vec{y}$ is called a \emph{counter update} as we refer to $k$ dimensions of a VASS as \emph{counters}.
We often omit the dimension in VASS and write $\VASS(\Z),\VASS(\N)$  if a definition or an argument is uniform w.r.t. the dimension.

We define the size of a VASS in a standard way assuming binary encoding of counter updates. 
Formally, the size of a VASS $\tuple{Q, Q_0, \delta}$ is defined as $|Q| + \sum_{(q,q',\vec{y}) \in \delta} \mathrm{len}(\vec{y})$, where $\mathrm{len}(\vec{y})$ is the length of the binary representation of $\vec{y}$.

\subparagraph*{Configurations and computations.} 
A \emph{configuration} of a $\VASS(\Z,k)$ $\vass$ is a pair from $Q \times \Z^k$, which consists of a state and a valuation of the counters.
A \emph{computation} of $\vass$ is an infinite sequence $\comp$ of configurations such that 
(a)~$\comp[0] \in Q_0 \times \{\vec{0}\}$, and 
(b)~for every $i \geq 0$, there exists $(q,q',\vec{y}) \in \delta$ such that 
$\comp[i] = (q,\vec{x})$ and $\comp[i+1] = (q',\vec{x}+\vec{y})$.
Note that, without loss of generality, we assume that the initial counter valuation is $\vec{0}$. We can encode any initial configuration in the VASS itself.

A computation of a $\VASS(\N,k)$ $\vass$ is a computation $\comp$ of $\vass$ considered as a $\VASS(\Z,k)$ such that 
the values of all counters are non-negative, i.e., for all $i$ we have $\comp[i] \in Q \times \N^k$.
Transitions of a $\VASS(\N)$ that make the value of some counter negative are disabled. 

We call a finite sequence $\compFin$ a \emph{subcomputation} of a $\VASS(\Z,k)$  (resp., $\VASS(\N,k)$) $\vass$, if it satisfies condition (b), i.e., all configurations are consistent with some transitions of $\vass$, and all configurations belong to $Q \times \Z^k$ (resp., $Q \times \N^k$).

\subparagraph*{Paths and cycles.} 
A path $\Path = (q_0,q'_0,\vec{y}_0), (q_1,q_1',\vec{y}_1), \ldots$ in a $\VASS(\Z)$ (resp., $\VASS(\N)$) $\vass$ is a  (finite or infinite) sequence of transitions (from $\delta$)
such that for all $0 \leq i < |\Path|$ we have $q'_i = q_{i+1}$. 
A finite path $\Path$ is a cycle if $\Path = (q_0,q'_0,\vec{y}_0), \ldots, (q_m,q'_m,\vec{y}_m)$ and $q_0 = q'_m$. 
Every computation in a $\VASS(\Z)$ (resp., $\VASS(\N)$) corresponds to the unique infinite path.
Conversely, every infinite path in a $\VASS(\Z)$ $\vass$ starting with $q_0 \in Q_0$ defines a computation in $\vass$. 
However, if $\vass$ is a $\VASS(\N,k)$, some paths do not correspond to valid computations due to non-negativity restriction posed on the counters. 

\subparagraph*{Cycle characteristics.}
For a path $\Path$ we define $\gain(\Path)$ as the vector of total counter change upon $\Path$.
Formally, for $\Path$ of length $n$ with counter updates $\vec{y}_1, \ldots, \vec{y}_{n}$ we define $\gain(\Path) = \sum_{i=1}^{n} \vec{y}_i$.

\subsection{Probabilistic semantics}

\subparagraph*{Markov chains.}
A \emph{Markov chain} is a tuple $\tuple{\Sigma,Q,Q_0,\delta,P,\mu}$ such that 
(1)~$\Sigma$ is a (finite) set of labels, 
(2)~$Q$ is a (finite) set of states, 
(3)~$Q_0$ is a set of initial states, 
(4)~$\delta \subseteq Q \times Q \times \Sigma$ is a transition relation, 
(5)~$P \colon \delta \to (0,1]$ is a probability distribution over transitions such that
for every $s \in S$ we have $\sum_{s' \in S,a \in \Sigma} p(s,s',a) = 1$, and
(6)~$\mu \colon Q_0 \to [0,1]$ is an initial distribution such that $\sum_{q \in Q_0} \mu(q) = 1$.

\subparagraph*{Probability measures defined by Markov chains.}
For a finite path $\Path$ in a Markov chain $\markov$, we define the probability of $\Path$, denoted by $\prob_{\markov}(\Path)$,
as the product of probabilities of transitions along $\Path$. 
For any $n > 0$, the probability $\prob_{\markov}(\cdot)$
is indeed a probability measure over paths of length $n$.
We extend this probability measure to infinite paths in the standard fashion. 
Let $X$ be the set of all infinite paths in $\markov$. 
For a basic open set $\Path \cdot X$, which is the set of all paths with the common prefix $\Path$, 
we define $\prob_{\markov}(\Path \cdot X)=\prob_{\markov}(\Path)$, and then the 
probability measure over infinite paths defined by $\markov$ is the unique
extension of the above measure (by Carath\'{e}odory's extension theorem~\cite{feller}).
We will denote the unique probability measure defined by $\markov$ as $\prob_{\markov}$.

\subparagraph*{Probabilistic VASS.} Probabilistic $\VASS$ generalize both $\VASS$ and Markov chains.
A~probabilistic $\VASS$ is a VASS, in which transitions are labeled with probabilities. 
It can be also considered to be an infinite-state Markov chain
over the set of states $Q \times \Z^k$ (resp., $Q \times \N^k$) and $\Sigma$ is a singleton.
Formally, a probabilistic VASS is a tuple $\vass = \tuple{Q, Q_0, \delta, P, \mu}$ such that 
(1)~$\tuple{Q, Q_0, \delta}$ is a VASS ($\VASS(\Z)$ or $\VASS(\N)$),
(2)~$P \colon \delta \to (0,1]$ is the probability distribution over transitions, which for every $q \in Q$ satisfies $\sum_{(q,q',\vec{y}) \in \delta} P(q,q',\vec{y}) = 1$, and
(3)~$\mu \colon Q_0 \to [0,1]$ is the initial distribution, which satisfies $\sum_{q \in Q_0} \mu(q) = 1$.

\subparagraph*{Probability measures defined by probabilistic VASS.}
A probabilistic $\VASS(\Z)$ (resp. $\VASS(\N)$) defines the probability measure over its computations. 
First, a probabilistic $\VASS(\Z)$ (resp., $\VASS(\N)$) $\vass$ defines the probability measure over its infinite paths
in the same way as a Markov chain does.
In $\VASS(\Z)$, every path corresponds to a computation and hence the probability measure over infinite paths carries over to computations.

\begin{itemize}
\item We define $\prob_{\vass}$ as the probability measure on computations carried over from infinite paths.
\end{itemize}

However, in $\VASS(\N)$ some paths may not correspond to valid computations. For that reason, defining the probability 
measure over computations poses difficulties~\cite{DBLP:conf/lics/BrazdilKKN15}. We consider two possible solutions: the \emph{strict} and the \emph{relaxed} semantics.
\begin{itemize}
    \item Under the strict semantics, we require that all paths correspond to valid computations and then 
we define the probability measure $\prob_{\vass}^s$ over computations as in the $\VASS(\Z)$ case. 

    \item Under the relaxed semantics, we require the set of paths corresponding to valid computations to have a non-zero
probability, and we define the probability measure $\prob_{\vass}^r$ over computations as the conditional probability under the condition
being the set of all paths that correspond to valid computations.
\end{itemize}

\subparagraph*{Random computations.}
To indicate that we consider a computation picked at random, we denote by $\randComp$ computations considered as random events.

\begin{remark} Under the strict semantics we require that every path corresponds to a valid computation, i.e., no counter gets a negative value. 
Note that relaxing \emph{all} to \emph{almost all} (i.e., with probability $1$) gives us the same notion.
Being a valid computation is a safety property and hence if the set of paths corresponding to valid computations
has probability $1$, then it is the set of all paths.
\end{remark}

\section{Problems}

In this section, we define \emph{the multi-dimensional average problem} 
and
\emph{the expected multi-dimensional average problem}, which we study in this paper.
We define the averages over selected positions; averages are parametrized by a set of states $S$, called \emph{selected states}, 
which determines meaningful configurations over which we compute the average, while skipping other configurations. 
This allows us to specify properties based on desired events (from $S$) rather than steps. 

\subparagraph*{Averages and limit-averages over selecting states.}
Let $\vass = \tuple{Q,Q_0,\delta}$ be a $\VASS(\Z,k)$ (resp., $\VASS(\N,k)$) and $S \subseteq Q$ be a set of \emph{selecting states}.
Fix a counter $i \in \{1, \ldots, k\}$.
For a finite subcomputation $\compFin$ of $\vass$, which contains at least one configuration from $S \times \Z^k$ (resp., $S \times \N^k$), 
we define the \emph{average value of counter $i$ (over $S$)}, denoted by $\avgS{S}^i(\compFin)$, as 
the average over values of counter $i$ over configurations with  the state belonging to $S$, i.e.,
we first pick a subsequence $(s_1, \vec{x}_1), \ldots, (s_m, \vec{x}_m)$ consisting of all configurations $(s,\vec{x})$ such that $s \in S$, and then take the average of the values of counter $i$:
\( 
 \avgS{S}^i(\compFin) = \frac{1}{m} \sum_{j=1}^m \vec{x}_j[i].
\)
If $\compFin$ has no configurations with states from $S$, then $\avgS{S}^i(\compFin)$ is undefined.
For an infinite sequence $\comp$ of configurations, which contains infinitely many configurations from $S \times \Z^k$ (resp., $S \times \N^k$), 
we define 
the \emph{limit-average value of the counter $i$ (over $S$)}, denoted by $\limavgS{S}^i(\comp)$, as 
\(
\limavgS{S}^i(\comp) = \liminf_{k \to \infty} \avgS{S}^i(\comp[0,k-1]).
\)
If $\comp$ does not contain infinitely many configurations from  $S \times \Z^k$ (resp., $S \times \N^k$), then $\limavgS{S}^i(\comp)$ is undefined.

\subparagraph*{Multi-dimensional averages and limit-averages over selecting states.}
We extend averages and limit-averages to multiple dimensions.
Let $\vec{S} = (S[1], \ldots, S[k])$ be a $k$-tuple of the subsets of $Q$.
For a (finite) subcomputation $\compFin$ and an (infinite) computation $\comp$, we define 
\[ 
\begin{split}
\avgS{\vec{S}}(\compFin) &= (\avgS{S[1]}^1(\compFin), \ldots, \avgS{S[k]}^k(\compFin)) \\
\limavgS{\vec{S}}(\comp) &= (\limavgS{S[1]}^1(\comp), \ldots, \limavgS{S[k]}^k(\comp))
\end{split}
\]
if all their components are defined. 
If any component of $\avgS{\vec{S}}(\compFin)$ (resp., $\limavgS{\vec{S}}(\comp)$) is undefined, the whole vector is undefined.

\begin{definition}[The multi-dimensional average problem for VASS]
Given a $\VASS(\N,k)$ (resp., $\VASS(\Z,k)$) $\vass$, $\vec{S} \in (2^Q)^k$ and $\thresholds \in \Q^k$, 
the \textbf{(multi-dimensional) average} problem asks whether there 
exists a computation $\comp$ such that $\limavgS{\vec{S}}(\comp)$ is defined and $\limavgS{\vec{S}}(\comp) \leq \thresholds$, i.e., 
the limit-averages of counter values over $\vec{S}$ are component-wise bounded by $\thresholds$.
\end{definition}

\subparagraph*{Expected limit-averages over selecting states.}
Consider a probabilistic $\VASS(\Z,k)$ (resp., $
\VASS(\N,k)$), which defines a probability measure $\prob_{\vass}$ (resp., $\prob_{\vass}^s$  
or $\prob_{\vass}^r$) over its computations. 
Let $\vec{S} = (S[1], \ldots, S[k])$ be a $k$-tuple of the subsets of $Q$.
The function $\randComp \mapsto \limavgS{S[i]}^i(\randComp)$ is a random variable w.r.t. $\prob_{\vass}$ (resp., $\prob_{\vass}^s$  
or $\prob_{\vass}^r$)  and 
we define $\expected_{\vass}(\limavgS{S[i]}^i)$ as the expected value of this random variable.
If the set of computations $\randComp$, at which $\limavgS{S[i]}^i(\randComp)$ is undefined, has a non-zero probability, 
then the expected value is undefined as well.
We extend the expectation to vectors and define the expected multi-dimensional limit-average as 
\[\expected_{\vass}(\limavgS{\vec{S}}) = (\expected_{\vass}(\limavgS{S[1]}^1), \ldots, \expected_{\vass}(\limavgS{S[k]}^k)).
\]
As above, the expected value $\expected_{\vass}(\limavgS{\vec{S}})$ is defined only if all components are defined.  

\begin{definition}[The expected (multi-dimensional) average problem for VASS]
Given a probabilistic VASS $\vass$, $\vec{S} \in (2^Q)^k$, 
 the \textbf{expected multi-dimensional average} problem asks to compute
the expected limit-averages over $\vec{S}$, i.e., $\expected_{\vass}(\limavgS{\vec{S}})$.
\end{definition}

\begin{remark}
In all complexity results for the multi-dimensional average and the expected multi-dimensional average problems, we consider VASS where the counter updates are encoded in binary.
\end{remark}

\section{Results on integer-valued VASS}

\subsection{The multi-dimension average problem}

Consider a $\VASS(\Z,k)$ $\vass = \tuple{Q,Q_0, \delta}$, a vector $\vec{S} \in (2^{Q})^k$ and thresholds
$\thresholds \in \Q^k$. For simplicity, we assume that $Q_{0} = \{q_0\}$ and hence $(q_0,\vec{0})$
is the initial configuration.

We present sufficient and necessary conditions for the existence of a computation $\comp$ with
$\limavgS{\vec{S}}(\comp) \leq \thresholds$. 
These conditions are expressed in terms of simple cycles in $\vass$, i.e., 
they stipulate  that for each counter $i$ there exist 
(a)~a simple cycle that can be \emph{iterated} 
to ensure that limit average infimum is consistent with the threshold $\thresholds[i]$, and 
(b)~a path to \emph{access} this cycle, and then to switch back to another cycle. 
These conditions can be check in $\NP$. 
We present main ideas assuming that $\vass$ is strongly connected. 

Assume that $\vass$ is strongly connected, i.e., it is strongly connected as a labeled graph.
We distinguish two types of counters based on their behavior in $\vass$: \emph{bounded} and \emph{unbounded}. 
We first assume that for every counter $i$ there is a cycle $\cycleI_i$ such that iterating this cycle 
decreases this counter's value, i.e., $\gain(\cycleI_i)[i] < 0$.
In such a case all counters are \emph{unbounded} and for any $\thresholds$ there exists a computation $\comp$ such that  $\limavgS{\vec{S}}(\comp) \leq \thresholds$.

\subparagraph*{The all-unbounded case.}
We assume that all counters are unbounded. Fix some $\thresholds \in \Q^k$. 
We construct  $\comp$ such that  $\limavgS{\vec{S}}(\comp) \leq \thresholds$ by interleaving strategies for each counter $i$ to 
make its partial average below $\thresholds[i]$. 
More precisely, we define the path $\Path$ of the form
\[
 \Path = \finPath^{1}_1 \finPath^{1}_2 \ldots \finPath^{1}_k \finPath^{2}_1 \ldots \finPath^{2}_k \ldots
\]
such that for every prefix $\finPath^{1}_1 \ldots \finPath^{j}_i$ of $\Path$,
the subcomputation $\compFin^{j}_{i}$ corresponding to that prefix satisfies $\avgS{S[i]}(\compFin^j_i) \leq \thresholds[i]$, i.e.,
 the partial average over $S[i]$ is bounded by $\thresholds[i]$. 
 We can construct such $\Path$ as follows. Suppose that a prefix of $\Path$ has been defined as above, and 
 we need to construct $\finPath^{j}_i$.
 There are two cases: if the cycle $\cycleI_i$ with $\gain(\cycleI_i)[i] < 0$ contains a selecting state from $S[i]$, then $\finPath^j_i = (\cycleI_i)^m$ for some large $m$, i.e., 
 we iterate $\cycleI_i$ long enough such that the average of the whole prefix computation is below $\thresholds[i]$.
 If $\cycleI_i$ does not contain any state from $S[i]$, then there exists $\cycleI_i'$ that contains a selecting state and 
$\gain(\cycleI_i')[i] < 0$. Indeed, let $\cycleD$ be a cycle from the initial state of $\cycleI_i$ to itself that contains a selecting state.
Then, $\cycleD \cdot \cycleI_i^N$ contains a selecting state and  $\gain(\cycleD \cdot \cycleI_i^N)[i] < 0$ for some $N$.

 Now, let $\comp$ be the computation corresponding to $\Path$. 
 For every counter $i$ there are infinitely many positions 
 at which the partial average over $S[i]$ at most  $\thresholds[i]$ and 
 hence $\limavgS{\vec{S}}(\comp) \leq \thresholds$.

\subparagraph*{The some-bounded case.}
Assume that for a counter $j$, there is no cycle such that iterating it decreases the value of counter $j$. 
In other words, for all cycles $c$ we have $\gain(c)[j] \geq 0$. We classify such a counter as \emph{bounded}. 
It is clearly lower bounded and for the limit average to be finite its has to be upper bounded. 
In consequence, in any computation $\comp$ with finite limit-average, all cycles $c$ that occur infinitely often  satisfy $\gain(c)[j] = 0$. 
This in turn restricts the set of cycles that can appear infinitely often in the considered paths, 
which makes other counters \emph{bounded}. 
We iterate this process until we reach a fixed point $\Bounded$, which is the set of all bounded counters.
The complement of $\Bounded$, denoted by $\UnBounded$, is the set of \emph{unbounded} counters. 

Note that for each unbounded counter $i \in \UnBounded$ there is a cycle $\cycleI_i$ such that:
\begin{description}
\item[(U1)\label{cond:uOne}] we have $\gain(\cycleI_i)[i] < 0$ and it contains a selecting state, and
\item[(U2)\label{cond:uTwoMinus}] for each bounded counter $j \in \Bounded$, we have $\gain(\cycleI_i)[j] = 0$.
\end{description}
It follows that similarly to the \emph{all-unbounded} case, we can make sure that the partial averages of unbounded counters are arbitrarily low. 

The limit average of a bounded counter depends on its initial value.
Indeed, in the extreme case, if the value of a counter $i$ does not change in any transition, then it is bounded and in every computation
the limit average of counter $i$ is precisely its initial value. 
However, to characterize cycles that witness low limit-averages of bounded counters it is more convenient to refer to a configuration that occurs infinitely often rather than 
the initial configuration. 
Therefore, we consider a \emph{recurring configuration} $(s_0, \vec{x})$ that is: 
(a)~reachable from the initial configuration $(q_0, \vec{0})$,
(b)~there are infinitely many configurations $(s_0,\vec{y})$ such that $\vec{y}$ and $\vec{x}$ agree on bounded counters.
We now drop the strongly-connected assumption on $\vass$.

Observe that switching between cycles for different (bounded or unbounded) counters may affect values of bounded counters. 
Therefore, for a bounded counter $i \in \Bounded$ we require that there is a cycle $\cycleI_i$, which 
(a)~can be \emph{accessed} with an appropriate path, and 
(b)~its average together with the initial value are bounded by $\thresholds[i]$. 
To make it more precise:
there exist a cycle $\cycleI_i$ and paths $\connectA_i, \connectB_i$ such that
\begin{description}
\item[(B1)\label{cond:bOne}] we have $\avgS{S[i]}(\compFin) \leq \thresholds[i]$, where 
$\compFin$ is the subcomputation corresponding to the cycle $\cycleI_i$ starting from the configuration reached 
from $(s_0,\vec{x})$ over the path $\connectA_i$, and 
\item[(B2)\label{cond:bTwo}]  $\connectA_i, \connectB_i$ are from $s_0$ to some $s \in \cycleI_i$ and from the same $s$ to $s_0$  respectively, and 
for each bounded counter $j \in \Bounded$, 
we have $\gain(\cycleI_i)[j] = 0$ and $\gain(\connectA_i \connectB_i)[j] = 0$.
\end{description}

Finally, we require that for all unbounded counters $i \in U$ there exist access paths $\connectA_i, \connectB_i$ 
as in condition~\textbf{\nameref{cond:bTwo}}, i.e., paths $\connectA_i, \connectB_i$ 
satisfy:
\begin{description}
\item[(U3)\label{cond:uTwo}]  $\connectA_i, \connectB_i$ are from $s_0$ to some $s \in \cycleI_i$ and from the same $s$ to $s_0$  respectively, and 
for each bounded counter $j \in B$, we have $\gain(\connectA_i \connectB_i)[j] = 0$.   
\end{description}
Condition~\textbf{\nameref{cond:uTwo}} is necessary as otherwise, switching between cycles for unbounded cycles and
bounded cycles could change values of bounded counters.
Observe that conditions~\textbf{\nameref{cond:uTwoMinus}} and~\textbf{\nameref{cond:uTwo}} together are the same as~\textbf{\nameref{cond:bTwo}}.
We unify these conditions into a single one denoted~\textbf{(BU)}. 

\subparagraph*{A witness for $\limavgS{\vec{S}} \leq \thresholds$.} A \emph{witness for $\limavgS{\vec{S}} \leq \thresholds$} is a tuple consisting of 
(a)~a (recurring) configuration $(s_0, \vec{x})$ reachable from the initial configuration $(q_0, \vec{0})$, 
(b)~a partition of counters into $\Bounded$ and $\UnBounded$, and
(c)~cycles $\cycleI_i$ and access paths $\connectA_i, \connectB_i$, for all $i$, which all satisfy conditions \allConditions.

First, we show that the existence of a witness for $\limavgS{\vec{S}} \leq \thresholds$ is sufficient for the existence 
of a computation $\comp$ with $\limavgS{\vec{S}}(\comp) \leq \thresholds$.

\keyIdeas{}
Using a witness, we construct a computation $\comp$ satisfying  $\limavgS{\vec{S}}(\comp) \leq \thresholds$ in a similar way as in the \emph{all-unbounded} case. 
The only difference here is that we use access paths to switch between cycles for different counters so that 
we switch between the counters in the state $s_0$, 
where the values of bounded counters are the same as in the configuration $(s_0, \vec{x})$. Due to condition~\textbf{(BU)}, we do not require $\vass$ to be strongly connected.

In consequence, we have the following:

\begin{restatable}{lemma}{ConditionsSufficient}
Let $\vass$ be a $\VASS(\Z)$. If it  has a witness for $\limavgS{\vec{S}} \leq \thresholds$,
then there exists a computation $\comp$ such that $\limavgS{\vec{S}}(\comp) \leq \thresholds$.
\end{restatable}
\begin{proof}
Let $(q_0,\vec{0})$ be the initial configuration and $(s_0,\vec{x})$ be the recurrent configuration of the witness. 
We assume that $\Bounded$ is non-empty. 
Otherwise, the construction presented in the all-bounded case essentially works. 
The only difference is that we use paths $\connectA_i, \connectB_i$ to switch between cycles.

We define the path $\Path$ of the form
\[
 \Path = \finPath_{0} \finPath^{1}_1 \finPath^{1}_2 \ldots \finPath^{1}_k \finPath^{2}_1 \ldots \finPath^{2}_k \ldots
\]
such that for every prefix $\finPath_{0} \finPath^{1}_1 \ldots \finPath^{j}_i$ of $\Path$,
the precomputation $\compFin^{j}_{i}$ corresponding to that prefix satisfies:
\begin{enumerate}[(a)]
    \item $\avgS{S[i]}(\compFin^j_i) \leq \thresholds[i] + \frac{1}{j}$, and
    \item $\compFin^j_i$ terminates in $(s_0,\vec{y})$, where for every $i \in \Bounded$ we have 
$\vec{y}[i] = \vec{x}[i]$. 
\end{enumerate}
Having such a path $\Path$, consider the computation $\comp$ that corresponds to $\Path$. 
Observe that for every counter $i$, condition (a) implies that for every $\epsilon >0$ there are infinitely many 
positions $k$ such that the average $\avgS{S[i]}(\comp[1,k])$ is less than $\thresholds[i] + \epsilon$ and hence
$\limavgS{S[i]}(\comp) \leq \thresholds[i]$. It follows that $\limavgS{\vec{S}}(\comp) \leq \thresholds$.

Now, we discuss how to construct such $\Path$. 
First, $\finPath_{0}$ is a path that corresponds to a computation from $(q_0,\vec{0})$ to $(s_0,\vec{x})$.
Second, suppose that a prefix $\Path_1$ of $\Path$ has been defined as above, and we need to construct $\finPath^{j}_i$.
Observe that the already constructed subcomputation ends in $(s_0,\vec{y})$ such that for every $k \in \Bounded$ we have $\vec{y}[i] = \vec{x}[i]$. 

There exist paths $\connectA_i, \connectB_i$ and a cycle $\cycleI_i$ satisfying~\textbf{(BU)}, and 
~\textbf{\nameref{cond:uOne}} (if $i \in \UnBounded$) or~\textbf{\nameref{cond:bOne}} (if $i \in \Bounded$).
Consider $\finPath^{j}_i$ of the form $\connectA_i \cycleI_i^N \connectB_i$ from some $N>0$ fixed later. 
First, due to condition~\textbf{(BU)}, for every $k \in \Bounded$ we have $\gain(\connectA_i \connectB_i) = 0$ and $\gain(\cycleI_i) = 0$, and hence
  $\gain(\finPath^{j}_i) = 0$. It follows that \textbf{(b)} holds.

Now, to see that \textbf{(a)} holds for $N$ big enough we consider two cases.
If counter $i$ is bounded, then~\textbf{\nameref{cond:bOne}} implies that for $\compFin_N$ being the computation corresponding to 
$\Path_1 \connectA_i \cycleI_i^N \connectB_i$, the average of the part corresponding to $\cycleI_i^N$ is $x$, which is less or equal to $\thresholds[i]$.
Therefore, $\avgS{S[i]}(\compFin_N)$ tends to $x$ as $N \to \infty$ and hence there is $N$ such that  $\avgS{S[i]}(\compFin_N) \leq \thresholds[i] + \frac{1}{j}$.

If counter $i$ is unbounded, then $\gain(\cycleI_i)[i] < 0$ and $\cycleI_i$ contains a selecting state. It follows that the values of counter $i$ tend to $-\infty$, and hence $\avgS{S[i]}(\compFin_N)$ tends to $-\infty$ as $N \to \infty$. 
Therefore, there exists $N$ such that  $\avgS{S[i]}(\compFin_N) \leq \thresholds[i] + \frac{1}{j}$.
\end{proof}

We show that the existence of a witness for $\limavgS{\vec{S}} \leq \thresholds$ is necessary for the existence 
of a computation $\comp$ with $\limavgS{\vec{S}}(\comp) \leq \thresholds$. 

\begin{restatable}{lemma}{ConditionsNecessary}
For all $\VASS(\Z,k)$ $\vass$ and $\thresholds \in \Q^k$ the following holds: 
if there is a computation $\comp$ such that $\limavgS{\vec{S}}(\comp) \leq \thresholds$, 
then there exists a witness for $\limavgS{\vec{S}} \leq \thresholds$, which has a polynomial size in $|\vass| + |\thresholds|$.
\end{restatable}
\begin{proof}
Consider a computation $\comp$ such that $\limavgS{\vec{S}}(\comp) \leq \thresholds$ and 
let $\Path$ be the infinite path corresponding to $\comp$.
We decompose $\Path$ into simple cycles greedily always picking the first occurring simple cycle. 
Now, consider all simple cycles that occur infinitely often as well as all rotations of these cycles $D = \set{\cycleD_1, \ldots, \cycleD_m}$. 
Based on these cycles, we define $\Bounded$ as the set of counters $j$ such that for all cycles $\cycleD \in D$ we have
$\gain(\cycleD)[j] = 0$, and $\UnBounded= \{1,\ldots,k\} \setminus \Bounded$.

Let $s_0$ be a state that occurs infinitely often in $\Path$. 
Note that eventually, past some position $K$, all transitions belong to cycles from $D$. Therefore, 
we pick the first configuration $(s_0, \vec{x})$ past position $K$ and observe that 
for all successive configurations $(s_0, \vec{y})$, for every counter $i \in \Bounded$, 
the gain between these configurations is $0$ and hence $\vec{x}[i] = \vec{y}[i]$. 
The length of description of $(s_0, \vec{x})$ is unbounded, but we show at the end of the proof that
it can be chosen to be  polynomial in $|\vass|+|\thresholds|$.
First, we show that there is any witness.

Consider a counter $j$ such that all cycles $\cycleD \in D$ satisfy $\gain(\cycleD)[j] \geq 0$.
We observe that for all cycles $\cycleD \in D$ we have $\gain(\cycleD)[j] = 0$ and hence $j \in \Bounded$.
Indeed, if there is a cycle $\cycleD \in D$ that satisfies $\gain(\cycleD)[j] > 0$, then 
values of counter $j$ in $\comp$ tend to $\infty$ and hence $\limavgS{S[j]}(\comp) = \infty > \thresholds[j]$.
It follows that for ever $i \in \UnBounded$, there is a cycle $\cycleI_i$ such that $\gain(\cycleI_i)[i] < 0$.

\subparagraph{The unbounded-counter case.} Consider  $i \in \UnBounded$. 
Let $\tilde{\cycleI_i} \in D$ be such that $\gain(\tilde{\cycleI_i})[i] < 0$ and let $q$ be the first state of $\tilde{\cycleI_i}$. 
Observe that there exist cycles $\cycleF_1, \cycleF_2 \in D$ such that 
$\cycleF_1$ is from $s_0$ to itself and contains $q$, and 
$\cycleF_2$ is from $q$ to itself and  contains some selecting state from $S_i$.
Indeed, $q$ and $s_0$ occur infinite often in $\Path$. 
Consider disjoint cycles $\cycleE_1, \cycleE_2, \ldots$ each from $s_0$ to itself that contains $q$.
For each $\cycleE_l$ we remove from it iteratively simple cycles from $D$ such that the resulting $\cycleE_l'$ does not contain any simple cycle from $D$.
Observe that only finitely many $\cycleE_l'$ are non-empty as otherwise there would be another simple cycle that occurs infinitely often and does not belong to $D$.
Now, let $\cycleE_l$ be a cycle that can be decomposed into simple cycles from $D$. Let us remove iteratively simple cycles to leave the ends $s_0$ of $\cycleE_l$ and 
a single occurrence of $q$. The resulting cycle consists of one or two simple cycles from $D$.
The proof for $\cycleF_2$ is similar.

Now, for  $N = |\gain(\cycleF_2)[i]|+1$ we define
$\cycleI_i = \cycleF_2 \tilde{\cycleI_i}^N$. Then, $\gain(\cycleI_i^N )[i] < 0$ and $\cycleI_i$ contains a selecting state from $S[i]$.
Therefore, condition~\textbf{\nameref{cond:uOne}} holds.

The cycle $\cycleF_1$ can be decomposed into paths $\connectA_i, \connectB_i$, respectively  from $s_0$ to $q$, and from $q$ to $s_0$.
Furthermore, since $\cycleF_2, \cycleI_i, \connectA_i \connectB_i$ can be decomposed into cycles from $D$, 
then by definition of $\Bounded$, for all $j \in \Bounded$
we have $\gain(\cycleF_2)[j] = \gain(\cycleI_i)[j] = \gain(\connectA_i \connectB_i)[j] =0$ and hence $\gain(\cycleF_2 \cycleI_i^n)[j] = 0$.
Therefore, condition~\textbf{(BU)} holds.
Note that $\connectA_i, \connectB_i$ have the lengths bounded by $2 \cdot |\vass|$ and $\cycleI_i$ can be represented by the pair of cycles: 
$(\tilde{\cycleI_i}, \cycleF_2)$ of the length at most $2 \cdot |\vass|$.
Thus, all have polynomial-size representation.

\subparagraph*{The bounded-counter case.} Let  $i \in \Bounded$. Since every cycle $\cycleD \in D$ satisfies $\gain(\cycleD)[i] = 0$, from some position $K$ 
onwards the gain of each cycle is $0$ and hence the value of counter $i$ on any two positions past $K$ with the same state are the same.
Therefore, we associate with each state the value of counter $i$ and eliminate values of counters. 
Furthermore, we associate with each cycle $\cycleD \in D$ its average value over $S[i]$, which is uniquely defined. 
Finally, if for all cycles $\cycleD \in D$ the average exceeds $\thresholds[i]$, then 
$\limavgS{S[i]}(\comp) > \thresholds[i]$. 
Therefore, there exists a cycle $\cycleI_i \in D$ with the average value less or equal to $\thresholds[i]$.

Since $s_0$ occurs infinitely often, in particular it occurs past position $K$.
Therefore, we show as in the unbounded-counter case that
there exist $\connectA_i, \connectB_i$ such that $\connectA_i$ leads from $s_0$ to some state $q$ of $\cycleI_i$ and $\connectB_i$ from $q$ to $s_0$ such that
$\cycleI_i, \connectA_i, \connectB_i$ satisfy~\textbf{(BU)}. Finally, observe that
 $\cycleI_i, \connectA_i, \connectB_i$ satisfy~\textbf{\nameref{cond:uOne}}.  
Note that $\cycleI_i, \connectA_i, \connectB_i$ can be picked to have the lengths  at most $2\cdot |\vass|$.

\subparagraph*{A witness with polynomial recurrent configuration.}
We have shown that there exists a witness for $\limavgS{\vec{S}} \leq \thresholds$ with $(s_0,\vec{x})$. 
We show that there exists $\vec{z}$ such that 
(a)~the witness for $\limavgS{\vec{S}} \leq \thresholds$ with $(s_0,\vec{x})$ replaced by $(s_0,\vec{z})$ 
remains a witness for $\limavgS{\vec{S}} \leq \thresholds$, and
(b)~$\vec{z}$ has the binary representation of polynomial length in $|\vass| + |\thresholds|$.

For (a) we need to show that (i)~\textbf{\nameref{cond:bOne}} is satisfied with $(s_0, \vec{z})$, and 
(ii) $(s_0,\vec{z})$ is reachable from $(q_0, \vec{0})$.
Recall that \textbf{\nameref{cond:bOne}} states that for every $i \in \Bounded$, 
the subcomputation $\compFin$ corresponding to the cycle $\cycleI_i$ starting from the configuration reached 
from $(s_0,\vec{z})$ over the path $\connectA_i$ satisfies  $\avgS{S[i]}(\compFin) \leq \thresholds[i]$.
Note that the lengths of $\connectA_i, \cycleI_i$ are bounded by $2\cdot |\vass|$ (because $i \in \Bounded$) and hence there exists $\alpha_i$ with the binary representation of 
polynomial-length in $|\vass|$ such that $\avgS{S[i]}(\compFin) = \alpha_i + \vec{z}[i]$. 
Therefore, any $\vec{z}$ such that $\vec{z}[i] < -\alpha_i + \thresholds[i]$ for all $i \in \Bounded$ satisfies (i). 
For (ii) observe that reachable configurations in $\VASS(\Z)$ are semilinear sets~\cite{DBLP:conf/lics/BlondinFGHM15} 
represented by polynomial-size equations (where coefficients are given in binary). 
Therefore, we can find a vector  $\vec{z}$ satisfying (i) and (ii)
whose binary representation has polynomial length in $|\vass| + |\thresholds|$.
\end{proof}

Finally, a polynomial-size witness for  $\limavgS{\vec{S}}(\comp) \leq \thresholds$ can be non-deterministically picked and verified in polynomial time.
More precisely, in the definition of a witness for $\limavgS{\vec{S}} \leq \thresholds$ condition (a)~can be checked in $\NP$ as reachability for $\VASS(\Z)$ is $\NP$-complete~\cite{DBLP:conf/lics/BlondinFGHM15}, and conditions (b) and (c) can be check in polynomial time. In consequence, we have:

\begin{lemma}
The multi-dimensional average problem for $\VASS(\Z)$ is in \NP. 
\end{lemma}

For hardness of the multi-dimensional average problem, consider configuration-reachability  for $\VASS(\Z)$, which is $\NP$-complete.
Configuration-reachability is mutually reducible to coverability for $\VASS(\Z)$~\cite{haase2014integer}, which in turn is equivalent to dual coverability, i.e, 
the problem, given a $\VASS(\Z)$ and two configurations $(s,\vec{0})$ and $(t, \vec{x})$, decide whether 
there is a (finite) subcomputation from $(s,\vec{0})$  to some $(t, \vec{y})$, where $\vec{y} \leq \vec{x}$.
The dual coverability straightforwardly reduces to the multi-dimensional average problem as follows. 
We construct $\vass'$ from  $\vass$ by adding a fresh state $t^*$ and two transitions labeled with $\vec{0}$:
from $t$ to $t^*$ and a self-loop over $t^*$.
Observe that there is a subcomputation from $(s,\vec{0})$  to $(t, \vec{y})$ where $\vec{y} \leq \vec{x}$ in $\vass$ if and only if
there is a computation from $(s,\vec{0})$ that eventually reaches $t^*$ and the multi-dimensional limit-averages 
are bounded by $\vec{x}$. To enforce that a computation eventually reaches $t^*$, we use an additional counter that is $0$ in the configurations of 
$\vass$ and it changes to $-1$ upon moving to $t^*$. Requirement that the limit average of this counter is less or equal to $-1$ forces the computation to move to $t^*$.
In consequence, the dual coverability for $\VASS(\Z)$ reduces to the multi-dimensional average problem 
for $\VASS(\Z)$ and hence the latter problem is $\NP$-complete.

\begin{theorem}
The multi-dimensional average problem for $\VASS(\Z)$ is \NP-complete. 
\end{theorem}

\subsection{The expected average problem}
\label{s:expectedZonSCC}

Observe that the expected average problem for probabilistic $\VASS(\Z)$ is modular and each
dimension can be considered separately.
This follows from the fact that each path in a $\VASS(\Z)$ corresponds to a computation, which is not the case for
$\VASS(\N)$. 
Furthermore, in this problem we compute the expected value over all computations and hence it can be considered for each dimension separately.
Therefore, we consider VASS that are  single-dimensional. 
We first discuss the strongly-connected case and then generalize our results to all $\VASS(\Z)$.

\subsubsection{The strongly-connected case}
Let $\vass$ be a single-dimensional probabilistic $\VASS(\Z,1)$, which is strongly connected as a labeled graph. 
We additionally assume that it has a single initial configuration $(q_0, 0)$. The case of any initial distribution follows easily.
First, we define the expected gain of $\vass$, which corresponds to the expected trend of the counter.

\subparagraph*{The expected gain.} The graph of $\vass$ can be considered as a Markov chain and using standard methods we compute for each state $q$ its long-run frequency $x_q$~\cite{BaierBook,filar}. 
More precisely, the \emph{frequency} of $q$ in a subcomputation $\randComp[1,n]$ is the number of configurations with the state $q$ in $\randComp[1,n]$ divided by $n$. The Ergodic Theorem for Markov chains implies that with probability $1$ over a random computation $\randComp$,
for every state $q$, the frequency of $q$ in $\randComp[1,n]$ converges to $x_q$ as $n$ tends to infinity.
Based on frequencies  $x_q$ we define \emph{the expected gain} $\expectedGain$ as the expected counter update provided that 
the initial state $q$ is picked at random according to the frequencies $x_q$ and the outgoing transition is picked at random according to 
the distribution at $q$, that is:
\[
\expectedGain = \sum_{(q,q',y) \in \delta} x_q \cdot P(q,q',y) \cdot y
\]

\subparagraph*{The classification based on $\expectedGain$.} 
We show that if $\expectedGain$ is positive (resp., negative), then the limit-average is infinite (resp. minus infinity).
However, if the expected gain is zero there are two cases based on boundedness of configurations.
Either the gain of every cycle is actually zero, or 
cycles with a positive gain balance cycles with a negative gain so that the expected gain is zero. We discuss these cases below.

We say that a $\VASS(\Z,1)$ is \emph{totally bounded} if the gain of each cycle is zero. 
This property does not depend on the probability distribution over transitions and 
we extend it straightforwardly to probabilistic $\VASS(\Z,1)$.
Note that if a $\VASS(\Z,1)$ is strongly connected and totally bounded, then in each reachable configuration, the state uniquely determines the counter's value. 
Otherwise, there exists a cycle with a non-zero gain. 
This observation allows us to reduce the expected limit-average problem for such VASS to
computing the expected long-run reward for Markov chains~\cite[Chapter 10.5]{BaierBook}. 

Consider a $\VASS(\Z,1)$ with $\expectedGain$ being zero and at least one cycle with a non-zero gain. 
Observe that $\expectedGain$ being $0$ implies that there is at least on cycle with a positive gain and a cycle with a negative gain. 
Let us consider the simplest probabilistic VASS, which has a single state $q_0$ and two self loops labeled with $1$ and $-1$, both with probability $0.5$.  
The distribution of the counter's gain in $n$ transitions, denoted $S_n$, is related to the binomial distribution $B(n,0.5)$ in the following way: $S_n \sim 2\cdot B(n,0.5)-n$. 
It follows that with probability $1$ over a random computation $\randComp$, the counter in $\randComp$ is neither lower nor upper bounded.
Furthermore, we show that with probability $1$, a random computation $\randComp$ has two subsequences such that the averages 
on one sequence tend to $\infty$, and on the other tend to $-\infty$. To state this formally we define:
\[
\begin{split}
\limavginfS{{S}}(\comp) = \liminf_{k \to \infty} \avgS{S}(\compFin[1,k]) \\
\limavgsupS{{S}}(\comp) = \limsup_{k \to \infty} \avgS{S}(\compFin[1,k]) 
\end{split}
\]

Now, we present the lemma summarizing the above discussion.

\begin{restatable}{lemma}{Classification}
\label{l:probabilisticClassification}
Let $\vass$ be a strongly-connected probabilistic $\VASS(\Z,1)$.
One of the following conditions holds: 
\begin{enumerate}[(1)]
\item $\expectedGain > 0$, and $\limavginfS{S}(\randComp) = \limavgsupS{S}(\randComp) = \infty$ with probability $1$ (over $\randComp$),
\item $\expectedGain < 0$, and $\limavginfS{S}(\randComp) = \limavgsupS{S}(\randComp) = -\infty$ with probability $1$,
\item $\vass$ is totally bounded, and for some $x \in \Q$, with probability $1$ over $\randComp$ we have 
$\limavginfS{S}(\randComp) = \limavgsupS{S}(\randComp) = x$, and
\item $\expectedGain = 0$, $\vass$ is not totally bounded, and 
with probability $1$ over $\randComp$ we have 
$\limavginfS{S}(\randComp) = -\infty$ and $\limavgsupS{S}(\randComp) = \infty$.
\end{enumerate}
\end{restatable}

\begin{proof}[Proof (of (1) and (2) from Lemma~\ref{l:probabilisticClassification})]
Assume that $\expectedGain \neq 0$.
The Ergodic Theorem for Markov chains implies that with probability $1$ (over $\randComp$)
for every state $q$, the frequency of configurations with the state $q$ converges to $x[q]$.
For every transition $(q,q',y)$, the frequency of this transition converges to $x_q \cdot P(q,q', {y})$.
Now, we multiply the frequency of each transition $(q,q',y)$ by its update value $y$ and 
get the value of the counter in $\randComp[n]$ divided by $n$.
On the other hand, this value converges to $\expectedGain$ as $n$ tends to infinity.
It follows that with probability $1$ over $\randComp$ the counter's value at a position $n$ equals 
$\expectedGain \cdot n \pm o(n)$. Therefore, if $\expectedGain >0$  we have $\limavginfS{S}(\randComp) = \limavgsupS{S}(\randComp) = \infty$.
Similarly, if $\expectedGain < 0$, then $\limavginfS{S}(\randComp) = \limavgsupS{S}(\randComp) = -\infty$ with probability $1$.
\end{proof}

\begin{proof}[Proof (of (3) from Lemma~\ref{l:probabilisticClassification})]
Assume that $\vass$ is totally bounded. In every computation $\comp$ if there are two configurations with the same state $(s,x_1), (s,x_2)$, then the counter's value is the same $x_1 = x_2$. To see that, consider a subcomputation from $(s,x_1)$ to $(s,x_2)$ and let $\Path$ be the path that corresponds to that subcomputation.
Then, $0 = \gain(\Path) = x_2 - x_1$.
Furthermore, since $\vass$ is strongly connected and $(q_0, 0)$ is the initial configuration for all computations, then in all computations 
the state determines the value of the counter. 

It follows that we can eliminate the counter and consider $\vass$ as a Markov chain with the limit-average objective with silent moves~\cite{LMCS19}.
In a Markov chain with silent moves, transitions are weighted with rational numbers and a special value $\bot$, which is skipped in the computation of partial averages.
Similarly to Markov chains, in strongly-connected Markov chains with silent moves, the expected limit-average in the Markov chain is actually the limit-average of almost all paths and it is our value $x$. Moreover, the expected value can be computed in polynomial time~\cite{LMCS19}.   

More precisely, let $h(q)$ be the value of the counter in the state $q$. 
We define a Markov chain $\markov$ corresponding to $\vass$ as follows.
We define $\markov = \tuple{\set{a},Q,\{ q_0\},\delta',P',\mu'}$ such that $\delta'(q,q',a)$ holds if and only if $(q,q',x) \in \delta$ for some $x \in \Z$, 
$P'(q,q',a) = \sum_{x \in \Z} P(q,q',x)$, and $\mu'(q_0) = 1$.  We consider the weighted Markov chain with silent moves 
$\tuple{\markov, \cost}$ such that
$\cost \colon \delta' \to \Z \cup \{ \bot \}$ is defined as $\cost(q,q',a) = h(q)$, if $q \in S$ is a selecting state, and 
$\cost(q,q',a) = \bot$ (is silent) otherwise. Observe that 
for every computation $\comp$ and the corresponding $\Path$ (without counter updates), we have $\limavgS{S}(\comp)$ is precisely the
limit average of costs $\cost$ of transitions along $\Path$. Since for almost all paths $\Path$ in $\tuple{\markov, \cost}$, 
the limit average cost of $\Path$ is the expected cost $x$ of $\tuple{\markov, \cost}$,
almost all computations in $\vass$ have the limit-average equal to $x$.
\end{proof}

It remains to prove (4) from Lemma~\ref{l:probabilisticClassification}). 
We only show that  $\limavginfS{S}(\randComp) = -\infty$ holds with probability $1$ over $\randComp$, as the proof
of  $\limavgsupS{S}(\randComp) = \infty$ is symmetric.
Observe that in a strongly-connected $\VASS(\Z)$, the event $\limavginfS{S}(\randComp) = -\infty$ is a tail event.
Therefore, due to Kolmogorov's 0-1 law~\cite{feller} it has  either probability $0$ or $1$. 
In consequence, it suffices to show that it has a positive probability.

First, we show that with a positive probability $\limavginfS{S}(\randComp)$ is upper bounded. 
\begin{restatable}{lemma}{LimInfIsFinite}
\label{l:LimInfIsFinite}
 Consider a probabilistic $\VASS(\Z)$ $\vass$ as in (4) of  Lemma~\ref{l:probabilisticClassification}.
 There exist $c \in \Q$ and $\delta > 0$ such that
 $\limavginfS{S}(\randComp) < c$ holds with probability greater than $\delta$.
\end{restatable}
\begin{proof}
 Let $\vass'$ results from $\vass$ by assuming that the all states are initial, i.e., $Q_0 = Q$, 
 and the initial distribution over states $\mu$ coincides with the long-run frequencies of states, i.e., $\mu(q) = x_q$.
 
Suppose that $\limavginfS{S}(\randComp) = \infty$  with probability $1$ w.r.t. $\vass$. 
Then, it also holds with probability $1$ w.r.t. $\vass'$.
Then, the average counter value at the $n$-th position converges to $\infty$ ($\lim_{n \to \infty} \avgS{S}(\randComp[1,n]) = \infty$) 
with probability $1$ in $\vass'$.
Therefore, the expected average counter value up to position $n$, $\expected_{\vass'}(\avgS{S}(\randComp[1,n]))$, converges to $\infty$.
However, the expected gain is $0$, which implies that in $\vass'$, at every position $n$, the expected value of the counter (in $\vass'$) is $0$. 
It follows that $\expected_{\vass'}(\avgS{S}(\randComp[1,n]))$ is $0$. A contradiction.
\end{proof}

For $c \in \Q$,  we define $X_{c}$ as the set of computations $\comp$ such that $\limavginfS{S}(\comp) < c$.
Lemma~\ref{l:LimInfIsFinite} states that there are $c \in \Q$ and $\delta > 0$ such that $\prob(X_c) = \delta$.
We show that for every $d \in \Q$,  $\limavginfS{S}(\randComp) < d$ holds with probability at least $\delta$.

 \begin{restatable}{lemma}{MiniusInfinityIndeed}
 \label{l:MiniusInfinityIndeed}
 Consider  a probabilistic $\VASS(\Z)$ $\vass$ as in (4) of Lemma~\ref{l:probabilisticClassification}.
 Assume that $\prob(X_c) = \delta > 0$.
 Then, for every $d$ we have $\prob(\limavginfS{S}(\randComp) < d ) \geq  \delta$.
 \end{restatable}
\begin{proof}
The main idea is to prepend to computations from $X_{c}$ a subcomputation that decreases the initial counter's value to $a$.
Then, the limit infimum of averages is $c+a$. Furthermore, we show that the set of such finite paths has probability $1$.

More precisely,  consider a subcomputation $\compFin$ from $(q_0,0)$ to $(q_0,a)$ and $\comp \in X_{c}$. 
We define the \emph{join} of $\compFin$ and $\comp$, denoted by $\compFin \Join \comp$, 
as the computation consisting of first $\compFin$ and then $\comp[1,\infty]$ ($\comp$ with the first configuration
removed) with $a$ added to the counter of all following configurations of $\comp$.
 Observe that the join of $\compFin$ and $\comp$ is indeed a computation. 
 Moreover, the influence of the average of $\compFin$ on the whole computation diminishes 
 and hence $\limavginfS{S}(\compFin \Join \comp) = \limavginfS{S}(\comp) + a < c + a$. 

Let $Y$ be the set of (finite) subcomputations that start in $(q_0, 0)$ and terminate once they reach some configuration $(q_0,b)$ where $b < d-c$. 
Since $\vass$ is strongly connected and not totally bounded, 
almost surely a random computation $\randComp$ reaches a configuration $(q_0,b)$ where $b < d-c$. 
It follows that the set of all computations extending some subcomputation from $Y$ has probability $1$.
Therefore, the set of all joins of subcomputations from $Y$ with computations from $X_{c}$
has probability at least $\delta$ and all such computations $\compFin \Join \comp$ satisfy $\limavginfS{S}(\compFin \Join \comp) < d$, and hence
Lemma~\ref{l:MiniusInfinityIndeed} follows.
\end{proof}
     
\begin{proof}[Proof (of (4) from Lemma~\ref{l:probabilisticClassification})]
Lemma~\ref{l:MiniusInfinityIndeed} implies that the set of computations $\randComp$ such that 
$\limavginfS{S}(\randComp) = -\infty$ has a positive probability. 
Since $\limavginfS{S}(\randComp) = -\infty$ is a tail event in a strongly-connected $\VASS(\Z)$, Kolmogorov's 0-1 law~\cite{feller} implies that its probability is $1$, which concludes
the proof of Lemma~\ref{l:probabilisticClassification}.
\end{proof}

Lemma~\ref{l:probabilisticClassification} implies the following:
\begin{lemma}
The expected average problem for strongly-connected probabilistic $\VASS(\Z,1)$ can be solved in polynomial time.
\end{lemma}
\begin{proof}[Proof's ideas] 
Consider a strongly-connected probabilistic $\VASS(\Z,1)$ $\vass$.
We can compute frequencies $x_q$ of states of $\vass$ in polynomial time using standard methods~\cite[Chapter 10.5]{BaierBook}.
Having frequencies $x_q$, we can compute the expected gain $\expectedGain$ of $\vass$ in polynomial time from the definition.

Assume that $\expectedGain = 0$. We can check whether $\vass$ is not totally bounded by checking 
whether it has a cycle with a non-zero gain, which can be done in polynomial time.
Finally, if it is totally bounded, then each state of $\vass$ uniquely determines the value of each counter,
and we can eliminate the counters and label states with counter values. Therefore,
the problem of computing the long-run average of almost all  computations, denoted by $x$, reduces to computing the expected long-run reward a Markov chain with rewards, which can be done in polynomial time~\cite[Chapter 10.5]{BaierBook}.
In consequence, we can check all the conditions of Lemma~\ref{l:probabilisticClassification} in polynomial time and hence the result follows.
\end{proof}

\subsubsection{The general case}
Let $\vass$ be a probabilistic $\VASS(\Z,1)$. 
We show how to compute its expected limit-average in polynomial time. 
We identify all \emph{bottom} SCCs (BSCCs) of $\vass$ (where an SCC $B$ is \emph{bottom} if all states reachable from $B$ 
belong to $B$). 
If there is a BSCC that does not contain a state from $S$, then the expected limit-average is undefined. 
Assume that every BSSC contains a state from $S$  and consider the following cases:
\begin{itemize}
    \item If there are two BSCCs: (a)~one with a positive expected gain, and  (b)~the other with a negative expected gain, then the expected limit average is undefined. The expected value is undefined for random variables that attain $+\infty$ and $-\infty$ with 
a positive probability~\cite{feller}.
\item If there is a BSCC with a positive gain and every BSCC has (a)~a positive gain, or (b)~it has the zero gain and it is totally bounded, 
then the expected limit-average  is $\infty$.
\item If there is a BSCC with (a)~a negative gain, or (b)~the zero gain and not totally bounded, 
and every BSCC has a non-positive gain, 
then the expected limit-average  is $-\infty$.
\item If all BSCCs have the zero gain and are totally bounded, the expected limit-average is finite and we discuss below how to compute it.
\end{itemize}

First, we compute all BSCCs $B_1, \ldots, B_m$ of $\vass$.
We pick in each of these components an initial state $q_0^i$.
For each BSCC $B_i$ with its initial configuration $(q_0^i,0)$, we compute $x_i$, which is the expected limit-average in $B_i$.
As we observed before, if we join a subcomputation from $(q_0, 0)$ to $(q_0^i,y_i)$ and some computation from $(q_0^i,0)$  with the limit-average $x_i$, then the limit-average of the resulting computation is $x_i + y_i$. 
Therefore, for each state $q_0^i$ we compute the probability of reaching that state from the initial distribution, denoted $p_i$, 
and the expected counter's value $y_i$ upon reaching $q_0^i$, i.e., the conditional expected counter's value under the condition
that the state $q_0^i$ is reached. Probabilities $p_i$ can be computed using standard methods for Markov chains~\cite[Chapter 10.1]{BaierBook}. 
The values $y_i$ can be computed as well using standard methods for Markov chains with rewards~\cite[Chapter 10.5]{BaierBook}.
Observe that the expected limit-average of $\vass$ is given by the following formula:
\[
\expected_\vass(\limavgS{S}) = \sum_{i=1}^m p_i \cdot (y_i + x_i) 
\] 

Finally, as we discussed above, we can compute the expected value for each counter separately. In consequence we have the following:

\begin{theorem}
\label{th:expectedZ}
The expected average problem for probabilistic $\VASS(\Z)$ can be solved in polynomial time.
\end{theorem}
 
\section{Results on natural-valued VASS}

\subsection{The average problem in a single dimension}

We first study the average problem for single-dimensional $\VASS(\N,1)$. 
For the lower bound observe that the reachability problem for  $\VASS(\N,1)$, which is $\NP$-complete~\cite{haase2009reachability},
reduces to the average problem for $\VASS(\N,1)$. 
The reduction is straightforward and hence we omit it. 
To show the $\NP$ upper bound, we show the following:

\begin{restatable}{lemma}{LemmaOndDimNatural}
\label{l:LemmaOndDimNatural}
For all $\VASS(\N,1)$ $\vass$ the following holds:
there exists a computation $\comp$ with $\limavgS{S}(\comp) \leq \lambda$ if
and only if  there exist subcomputations $\compFin_0, \compFin_c$ such that:
 \begin{itemize}
    \item $\compFin_0$ is from $(q_0,0)$ to  $(s,x)$, where $x \leq \lambda$, and
    \item $\compFin_c$ is a cycle from $(s,x)$ to itself satisfying the following conditions:
    \begin{enumerate}[(a)]
      \item $\avgS{S}(\compFin_c) \leq  \lambda$,
      \item the number of configurations with selecting states in $\compFin_c$, i.e., configurations from $S \times \N$, is 
     $O(|S|\cdot|\lambda|^2)$, and 
      \item the value of the counter in each configuration of $\compFin_c$ from $S \times \N$ is 
      $O(|S|\cdot|\lambda|^2)$.
  \end{enumerate}
\end{itemize}
\end{restatable}
\begin{proof}
Observe that having $\compFin_0, \compFin_c$ as above, the computation $\comp = \compFin_0 (\compFin_c)^{\infty}$ is a valid computation
and it satisfies $\limavgS{S}(\comp) \leq \lambda$.

Conversely, assume that there is a computation $\comp$ with $\limavgS{S}(\comp) \leq \lambda$.
Consider $\epsilon > 0$ and pick a cycle subcomputation $\compFin$ from $\comp$ of the average value at most $\lambda + \epsilon$ of the minimal length (all shorter subcomputations have higher average). Such a cycle exists as there has to be a configuration $(s,y)$ with $y \leq \lambda + \epsilon$ that occurs infinitely often. Otherwise, $\limavgS{S}(\comp) \geq \lambda + \epsilon$. Then, we divide $\comp$ into cycles with ends with configuration $(s,y)$ and there has to be a cycle with the average value at most $\lambda + \epsilon$.

We show that this minimal $\compFin$ has few \emph{selecting configurations}, which are configurations with a selecting state.
Let $L$ be the number of selecting configurations in $\compFin$ with the counter's value at most $\lceil \lambda  \rceil$ and let
$H$ be the number of selecting configurations with the counters value at least  $\lceil \lambda  \rceil + 1$.
We lower the average if we replace the value of the configurations of the first type by $0$ and the second by $\lceil \lambda  \rceil + 1$ and get
\[
\frac{ (\lceil \lambda  \rceil + 1) H}{L+H} \leq \lambda + \epsilon
\]
and hence
\(
H \leq \frac{\lambda + \epsilon}{1-\epsilon} \cdot L
\)
Assuming that $\epsilon \leq 0.5$, we can bound $H \leq 2\cdot(\lceil \lambda  \rceil + 1) L$.

Now, we give a bound on $L$. Due to minimality assumption on $\compFin$, it cannot contain subcycles with the same properties. 
Suppose it has a subcycle $\tau$. Due to minimality assumption, the subcycle has the average value exceeding $\lambda + \epsilon$.
But then, $\compFin'$ obtained from $\compFin$ by removal of $\tau$ has a smaller average and a shorter length. A contradiction.
It follows that for every state $q$ and for every value $x \leq \lceil \lambda  \rceil$ there is at most one configuration $(q,x)$ in $\compFin$. 
Therefore, $L \leq (\lceil \lambda  \rceil +1) |S|$.
and hence 
\[
L+H \leq (\lceil \lambda  \rceil +1 )|S|  + 2\cdot (\lceil \lambda  \rceil +1)^2|S| \leq 2 \cdot |S| \cdot (\lceil \lambda  \rceil + 2)^2.
\]

As previously observed the minimal value of a configuration from $L$ is $0$ and from $H$ is $\lceil \lambda  \rceil +1$. 
Suppose that there is a single high value $B$ is $\compFin$ and all other values take the minimal possible value. Then, we get:
\[
\frac{B + (H-1)  \lceil \lambda  \rceil}{L+H} \leq \lambda 
\]
thus
\[
B \leq  \lceil \lambda  \rceil (L+1) \leq (\lceil \lambda  \rceil +1)^2 |Q|
\]
and that is the bound on the maximal value of a selecting configuration.
\end{proof}

We can check in $\NP$ whether there exist  $\compFin_0, \compFin_c$ satisfying the conditions from 
Lemma~\ref{l:LemmaOndDimNatural}.

\keyIdeas{} We non-deterministically pick all selecting configurations $(s_1, x_1), \ldots, (s_m,x_m)$ from $\compFin_c$.
Then, we check reachability from $(q_0, 0)$ to $(s_1,x_1)$, and 
for each $i <m$ reachability over non-selecting configurations
from $(s_{i},x_{i})$ to $(s_{i+1},x_{i+1})$, and from $(s_m,x_m)$ to $(s_1,x_1)$. 
All these reachability checks can be done in \NP.
Finally, we check $\frac{1}{m} \sum_{i=1}^m x_i \leq \lambda$.
All these checks can be done is \NP. The number of configurations $m$ as well as the size of each configuration is polynomially bounded due to Lemma~\ref{l:LemmaOndDimNatural}. In consequence, we have:

\begin{theorem}
The average problem for $\VASS(\N,1)$ is $\NP$-complete.
\end{theorem}

\subsection{The multi-dimension average problem}
We show that the (decision variant of the) multi-dimensional average problem for $\VASS(\N)$ is undecidable.
A related problem, called the average-value problem, has been studied in~\cite{DBLP:conf/concur/ChatterjeeHO19}. 
In that problem, the values of all counters in each configuration $(q,\vec{x})$ are aggregated into a single number, called \emph{the cost}, 
by computing dot-product of $\vec{x}$ and a cost vector $\vec{c}_q \in N^k$.
A cost vector $\vec{c}_q$ depends on the state $q$ in the configuration. The average-value problem asks whether
there exists a computation such that the limit average of costs is less or equal to a threshold $\lambda$.
The problem for $\VASS(\N)$ with threshold $0$ is undecidable~\cite[Theorem~24]{DBLP:conf/concur/ChatterjeeHO19}.
Threshold $0$ in $\VASS(\N)$ means that whenever a cost vector is non-zero at component $i$ (i.e, $\vec{c}_q[i] \neq 0$), then counter's $i$ value should be $0$. 
This constraint is expressed in the multi-dimensional average problem and hence we have:

\begin{theorem}
The decision variant of the  multi-dimensional average problem for $\VASS(\N)$ is undecidable.
\end{theorem}

\subsection{The expected multi-dimensional average problem}

We first study probabilistic $\VASS(\N)$ under the strict semantics and give the precise complexity. 
Next, we consider the relaxed semantics, where we have the exact complexity in the strongly-connected case
and a hardness result in the general case.

\subsubsection{Probabilistic natural-valued VASS under the strict semantics}
Consider a probabilistic $\VASS(\N)$ $\vass$ under the strict semantics.
To check whether every path of $\vass$ corresponds to a valid computation, we examine each counter $i$ separately and check whether it can reach a negative value from some initial configuration. This can be done in polynomial time with the standard reachability analysis. If a negative value for some counter is reachable, then the expected limit-average is undefined under the strict semantics.
Otherwise, every path in $\vass$ corresponds to a valid computation and we can consider $\vass$ as 
a $\VASS(\Z)$ as the non-negativity restriction is vacuous for $\vass$. Therefore, we apply Theorem~\ref{th:expectedZ} and 
compute the expected value for $\vass$. In consequence, we have:

\begin{theorem}
\label{th:expectedNstrict}
The expected average problem for probabilistic $\VASS(\N)$ under the strict semantics can be solved in polynomial time.
\end{theorem}

\subsubsection{Probabilistic natural-valued VASS under the relaxed semantics}

\subparagraph*{The finite strongly-connected case.} 
Consider a probabilistic $\VASS(\N)$ $\vass$, which is strongly connected. We show that if the expected limit-average is 
finite, then the strict and the relaxed semantics coincide.
We first assume that $\vass$ is single-dimensional. 
Using the classification from Lemma~\ref{l:probabilisticClassification} applied to $\vass$ considered as a $\VASS(\Z,1)$, we observe that:
\begin{itemize}
    \item If $\expectedGain < 0$ or $\expectedGain = 0$ and $\vass$ is not totally bounded, 
    then a random computation $\randComp$ 
    (under the $\VASS(\Z)$ semantics) satisfies $\limavginfS{S}(\randComp) = -\infty$, and hence it is not a valid computation of the $\VASS(\N)$. 
    Therefore, the expected limit-average under the relaxed semantics is undefined for $\vass$.

    \item If $\expectedGain > 0$, then $\limavginfS{S}(\randComp) = \infty$. Therefore, if
    the set of random computations $\randComp$ (under the $\VASS(\Z)$ semantics) that are also valid computations
    under the $\VASS(\N)$ semantics has a positive probability, then 
   the expected limit-average under the relaxed semantics is defined and infinite. Otherwise, it is undefined.

    \item If $\expectedGain = 0$ and $\vass$ is totally bounded, then (as we observe in Section~\ref{s:expectedZonSCC}) in each configuration, 
    the state uniquely determines the counters value. 
    Therefore, we can check whether counter values in all states are non-negative. 
    If this is the case, then all paths correspond to valid computations, the expected limit-average is defined and finite, 
    and we can compute it with Theorem~\ref{th:expectedNstrict}.
Otherwise, observe that in a strongly-connected $\vass$ every state is visited with probability $1$ and 
    hence the expected value is undefined.
\end{itemize}

Therefore, for the expected limit-average under the relaxed semantics to be defined and finite, the expected gain w.r.t. 
every counter has to be $0$ and it has to be totally bounded. Furthermore, we check for each counter independently whether 
every path corresponds to a valid computation in $\VASS(\N)$. Since we consider all paths, we can make these checks independently for all counters.
In consequence we have the following:

\begin{theorem}
\label{th:expectedNrelaxedSCC}
Deciding whether the expected limit-average is defined and finite over strongly-connected probabilistic $\VASS(\N)$ under the relaxed semantics 
can be solved in polynomial time. Furthermore, it if is it can be computed in polynomial time.
\end{theorem}

\subparagraph*{The general case.} We present only a hardness result. 
The coverability problem for  $\VASS(\N)$, which is $\EXPSPACE$-complete~\cite{lipton1976reachability}, reduces to 
(the decision version of) the expected limit-average problem for $\VASS(\N)$. 
The reduction is rather straightforward with minor technical difficulties (we need to ensure that the expected value of each counter is finite).
In consequence, we have:

\begin{restatable}{theorem}{ReductionToCoverability}
\label{th:HardnessCover}
The problem, given a probabilistic $\VASS(\N,k)$ $\vass$ under the relaxed semantics, $S \subseteq Q$ and $\vec{x} \in \Q^k$, decide whether 
$\expected_{\vass}(\limavgS{\vec{S}}) < \vec{x}$  is $\EXPSPACE$-hard. 
\end{restatable}
\begin{proof}
Consider a $\VASS(\N,k)$ $\vass$, an initial configuration $(s,\vec{x}_1)$, and a target configuration  $(t,\vec{x}_2)$.
Without loss of generality, we assume that $\vec{x}_1 = \vec{x}_2 = \vec{0}$ and the update of each counter is $-1,0,1$. 
We construct a probabilistic $\VASS(\N)$ $\vass^P$ based on $\vass$ by adding an additional counter $k+1$ and a sink state $r$,
which has only a single outgoing transition, which is a self-loop upon which counters do not change values, i.e., $(r,r,\vec{0})$.
We add a transition from $t$ to $r$ labeled with $\vec{0}$ and assign to it some positive probability.  
To make sure that the expected value of $\vass^P$ is defined and finite, we add to every state $\vass^P$ a transition to $r$ labeled with $\vec{1}$
with probability $\frac{1}{2}$.  We assign positive probabilities to the remaining transitions. 
Observe that a random computation reaches the sink $r$ with probability $1$, and the probability that it happens after more than $n$ steps is 
bounded by $(\frac{1}{2})^n$. The value of the counter after $n$ steps is at most $n$. Therefore, 
the expected limit-average of counters $1,  \ldots, k$ is bounded by $\sum_{j=1}^{\infty} (\frac{1}{2})^i \cdot i = 2$.
Upon reaching $r$ the counter $k+1$ has value $0$, if the previous state was $t$, and $1$ otherwise.  
Therefore, the expected limit-average is strictly less than $(2+\epsilon, \ldots, 2+\epsilon, 1)$ 
(for any $\epsilon > 0$) if and only if there is a  computation from $(s,\vec{0})$ to $(t, \vec{y})$ 
(with any $\vec{y}$) in $\vass$.   
\end{proof}

\subparagraph*{Remark.} The decidability of the problem from Theorem~\ref{th:HardnessCover} is open.

\bibliography{papers}

\begin{thebibliography}{10}

\bibitem{AAHMKT14}
Parosh~Aziz Abdulla, Mohamed~Faouzi Atig, Piotr Hofman, Richard Mayr,
  K.~Narayan Kumar, and Patrick Totzke.
\newblock Infinite-state energy games.
\newblock In {\em {CSL-LICS} 2014}, pages 7:1--7:10, 2014.
\newblock \href {https://doi.org/10.1145/2603088.2603100}
  {\path{doi:10.1145/2603088.2603100}}.

\bibitem{BaierBook}
Christel Baier and Joost{-}Pieter Katoen.
\newblock {\em Principles of model checking}.
\newblock {MIT} Press, 2008.

\bibitem{Bloem16}
Roderick Bloem, Swen Jacobs, Ayrat Khalimov, Igor Konnov, Sasha Rubin, Helmut
  Veith, and Josef Widder.
\newblock Decidability in parameterized verification.
\newblock {\em {SIGACT} News}, 47(2):53--64, 2016.
\newblock \href {https://doi.org/10.1145/2951860.2951873}
  {\path{doi:10.1145/2951860.2951873}}.

\bibitem{DBLP:conf/lics/BlondinFGHM15}
Michael Blondin, Alain Finkel, Stefan G{\"{o}}ller, Christoph Haase, and Pierre
  McKenzie.
\newblock Reachability in two-dimensional vector addition systems with states
  is pspace-complete.
\newblock In {\em {LICS} 2015}, pages 32--43. {IEEE} Computer Society, 2015.
\newblock \href {https://doi.org/10.1109/LICS.2015.14}
  {\path{doi:10.1109/LICS.2015.14}}.

\bibitem{BCKNVZ18}
Tom{\'{a}}s Br{\'{a}}zdil, Krishnendu Chatterjee, Anton{\'{\i}}n Kucera, Petr
  Novotn{\'{y}}, Dominik Velan, and Florian Zuleger.
\newblock Efficient algorithms for asymptotic bounds on termination time in
  {VASS}.
\newblock In {\em {LICS} 2018}, pages 185--194, 2018.
\newblock \href {https://doi.org/10.1145/3209108.3209191}
  {\path{doi:10.1145/3209108.3209191}}.

\bibitem{DBLP:conf/lics/BrazdilKKN15}
Tom{\'{a}}s Br{\'{a}}zdil, Stefan Kiefer, Anton{\'{\i}}n Kucera, and Petr
  Novotn{\'{y}}.
\newblock Long-run average behaviour of probabilistic vector addition systems.
\newblock In {\em {LICS} 2015}, pages 44--55, 2015.
\newblock \href {https://doi.org/10.1109/LICS.2015.15}
  {\path{doi:10.1109/LICS.2015.15}}.

\bibitem{CHO16b}
Krishnendu Chatterjee, Thomas~A. Henzinger, and Jan Otop.
\newblock Nested weighted limit-average automata of bounded width.
\newblock In {\em {MFCS} 2016}, pages 24:1--24:14, 2016.
\newblock \href {https://doi.org/10.4230/LIPIcs.MFCS.2016.24}
  {\path{doi:10.4230/LIPIcs.MFCS.2016.24}}.

\bibitem{CHO16a}
Krishnendu Chatterjee, Thomas~A. Henzinger, and Jan Otop.
\newblock Quantitative monitor automata.
\newblock In {\em {SAS} 2016}, pages 23--38, 2016.
\newblock \href {https://doi.org/10.1007/978-3-662-53413-7\_2}
  {\path{doi:10.1007/978-3-662-53413-7\_2}}.

\bibitem{DBLP:conf/concur/ChatterjeeHO19}
Krishnendu Chatterjee, Thomas~A. Henzinger, and Jan Otop.
\newblock Long-run average behavior of vector addition systems with states.
\newblock In {\em {CONCUR} 2019}, pages 27:1--27:16, 2019.
\newblock \href {https://doi.org/10.4230/LIPIcs.CONCUR.2019.27}
  {\path{doi:10.4230/LIPIcs.CONCUR.2019.27}}.

\bibitem{LMCS19}
Krishnendu Chatterjee, Thomas~A. Henzinger, and Jan Otop.
\newblock Quantitative automata under probabilistic semantics.
\newblock {\em Logical Methods in Computer Science}, 15(3), 2019.
\newblock \href {https://doi.org/10.23638/LMCS-15(3:16)2019}
  {\path{doi:10.23638/LMCS-15(3:16)2019}}.

\bibitem{CV17a}
Krishnendu Chatterjee and Yaron Velner.
\newblock The complexity of mean-payoff pushdown games.
\newblock {\em J. {ACM}}, 64(5):34:1--34:49, 2017.
\newblock \href {https://doi.org/10.1145/3121408} {\path{doi:10.1145/3121408}}.

\bibitem{CV17b}
Krishnendu Chatterjee and Yaron Velner.
\newblock Hyperplane separation technique for multidimensional mean-payoff
  games.
\newblock {\em J. Comput. Syst. Sci.}, 88:236--259, 2017.
\newblock \href {https://doi.org/10.1016/j.jcss.2017.04.005}
  {\path{doi:10.1016/j.jcss.2017.04.005}}.

\bibitem{DBLP:conf/stoc/CzerwinskiLLLM19}
Wojciech Czerwinski, Slawomir Lasota, Ranko Lazic, J{\'{e}}r{\^{o}}me Leroux,
  and Filip Mazowiecki.
\newblock The reachability problem for petri nets is not elementary.
\newblock In {\em {STOC} 2019}, pages 24--33, 2019.
\newblock \href {https://doi.org/10.1145/3313276.3316369}
  {\path{doi:10.1145/3313276.3316369}}.

\bibitem{DKO13:conc-verification-vass}
Emanuele D'Osualdo, Jonathan Kochems, and C.{-}H.~Luke Ong.
\newblock Automatic verification of erlang-style concurrency.
\newblock In {\em {SAS} 2013}, pages 454--476, 2013.
\newblock \href {https://doi.org/10.1007/978-3-642-38856-9\_24}
  {\path{doi:10.1007/978-3-642-38856-9\_24}}.

\bibitem{Esparza:PN}
Javier Esparza.
\newblock Decidability and complexity of petri net problems—an introduction.
\newblock {\em Lectures on Petri nets I: Basic models}, pages 374--428, 1998.

\bibitem{EN94}
Javier Esparza and Mogens Nielsen.
\newblock Decidability issues for petri nets - a survey.
\newblock {\em Bull. {EATCS}}, 52:244--262, 1994.

\bibitem{feller}
W.~Feller.
\newblock {\em An introduction to probability theory and its applications}.
\newblock Wiley, 1971.

\bibitem{FMWDR17:component-based-synthesis}
Yu~Feng, Ruben Martins, Yuepeng Wang, Isil Dillig, and Thomas~W. Reps.
\newblock Component-based synthesis for complex apis.
\newblock In {\em POPL 2017}, pages 599--612, New York, NY, USA, 2017. ACM.
\newblock URL: \url{http://doi.acm.org/10.1145/3009837.3009851}, \href
  {https://doi.org/10.1145/3009837.3009851}
  {\path{doi:10.1145/3009837.3009851}}.

\bibitem{filar}
Jerzy Filar and Koos Vrieze.
\newblock {\em Competitive Markov decision processes}.
\newblock Springer, 1996.

\bibitem{GM12:asynchronous-verification-TOPLAS}
Pierre Ganty and Rupak Majumdar.
\newblock Algorithmic verification of asynchronous programs.
\newblock {\em ACM Trans. Program. Lang. Syst.}, 34(1):6:1--6:48, May 2012.
\newblock URL: \url{http://doi.acm.org/10.1145/2160910.2160915}, \href
  {https://doi.org/10.1145/2160910.2160915}
  {\path{doi:10.1145/2160910.2160915}}.

\bibitem{haase2014integer}
Christoph Haase and Simon Halfon.
\newblock Integer vector addition systems with states.
\newblock In {\em {RP} 2014}, pages 112--124, 2014.
\newblock \href {https://doi.org/10.1007/978-3-319-11439-2\_9}
  {\path{doi:10.1007/978-3-319-11439-2\_9}}.

\bibitem{haase2009reachability}
Christoph Haase, Stephan Kreutzer, Jo{\"{e}}l Ouaknine, and James Worrell.
\newblock Reachability in succinct and parametric one-counter automata.
\newblock In {\em {CONCUR} 2009}, pages 369--383, 2009.
\newblock \href {https://doi.org/10.1007/978-3-642-04081-8\_25}
  {\path{doi:10.1007/978-3-642-04081-8\_25}}.

\bibitem{KKW10:dynamic-cutoff-detection}
Alexander Kaiser, Daniel Kroening, and Thomas Wahl.
\newblock Dynamic cutoff detection in parameterized concurrent programs.
\newblock In {\em {CAV} 2010}, pages 645--659, 2010.
\newblock \href {https://doi.org/10.1007/978-3-642-14295-6\_55}
  {\path{doi:10.1007/978-3-642-14295-6\_55}}.

\bibitem{KKW12:coverability-proof-minim}
Alexander Kaiser, Daniel Kroening, and Thomas Wahl.
\newblock Efficient coverability analysis by proof minimization.
\newblock In Maciej Koutny and Irek Ulidowski, editors, {\em CONCUR 2012},
  pages 500--515, Berlin, Heidelberg, 2012. Springer Berlin Heidelberg.
\newblock URL: \url{http://dx.doi.org/10.1007/978-3-642-32940-1_35}, \href
  {https://doi.org/10.1007/978-3-642-32940-1_35}
  {\path{doi:10.1007/978-3-642-32940-1_35}}.

\bibitem{KM69}
Richard~M. Karp and Raymond~E. Miller.
\newblock Parallel program schemata.
\newblock {\em J. Comput. Syst. Sci.}, 3(2):147--195, 1969.
\newblock \href {https://doi.org/10.1016/S0022-0000(69)80011-5}
  {\path{doi:10.1016/S0022-0000(69)80011-5}}.

\bibitem{Kosaraju82}
S.~Rao Kosaraju.
\newblock Decidability of reachability in vector addition systems (preliminary
  version).
\newblock In {\em Proceedings of the 14th Annual {ACM} Symposium on Theory of
  Computing, May 5-7, 1982, San Francisco, California, {USA}}, pages 267--281,
  1982.
\newblock \href {https://doi.org/10.1145/800070.802201}
  {\path{doi:10.1145/800070.802201}}.

\bibitem{DBLP:journals/tcs/Lambert92}
Jean{-}Luc Lambert.
\newblock A structure to decide reachability in petri nets.
\newblock {\em Theoretical Computer Science}, 99(1):79--104, 1992.
\newblock \href {https://doi.org/10.1016/0304-3975(92)90173-D}
  {\path{doi:10.1016/0304-3975(92)90173-D}}.

\bibitem{leroux2012vector}
J{\'{e}}r{\^{o}}me Leroux.
\newblock Vector addition systems reachability problem {(A} simpler solution).
\newblock In {\em Turing-100 - The Alan Turing Centenary, Manchester, UK, June
  22-25, 2012}, pages 214--228, 2012.
\newblock URL: \url{https://easychair.org/publications/paper/Blr}.

\bibitem{Ler18}
J{\'{e}}r{\^{o}}me Leroux.
\newblock Polynomial vector addition systems with states.
\newblock In {\em {ICALP} 2018}, pages 134:1--134:13, 2018.
\newblock \href {https://doi.org/10.4230/LIPIcs.ICALP.2018.134}
  {\path{doi:10.4230/LIPIcs.ICALP.2018.134}}.

\bibitem{lipton1976reachability}
Richard Lipton.
\newblock The reachability problem is exponential-space hard.
\newblock {\em Department of Computer Science, Yale University, Tech. Rep}, 62,
  1976.

\bibitem{DBLP:conf/stoc/Mayr81}
Ernst~W. Mayr.
\newblock {An Algorithm for the General Petri Net Reachability Problem}.
\newblock In {\em STOC 1981}, pages 238--246, 1981.
\newblock \href {https://doi.org/10.1145/800076.802477}
  {\path{doi:10.1145/800076.802477}}.

\bibitem{DBLP:conf/fsttcs/MichaliszynO17}
Jakub Michaliszyn and Jan Otop.
\newblock Average stack cost of b{\"{u}}chi pushdown automata.
\newblock In {\em {FSTTCS} 2017}, pages 42:1--42:13, 2017.
\newblock \href {https://doi.org/10.4230/LIPIcs.FSTTCS.2017.42}
  {\path{doi:10.4230/LIPIcs.FSTTCS.2017.42}}.

\bibitem{RACKOFF1978223}
Charles Rackoff.
\newblock The covering and boundedness problems for vector addition systems.
\newblock {\em Theoretical Computer Science}, 6(2):223 -- 231, 1978.
\newblock \href {https://doi.org/https://doi.org/10.1016/0304-3975(78)90036-1}
  {\path{doi:https://doi.org/10.1016/0304-3975(78)90036-1}}.

\bibitem{SZV14}
Moritz Sinn, Florian Zuleger, and Helmut Veith.
\newblock A simple and scalable static analysis for bound analysis and
  amortized complexity analysis.
\newblock In {\em {CAV} 2014}, pages 745--761, 2014.
\newblock \href {https://doi.org/10.1007/978-3-319-08867-9\_50}
  {\path{doi:10.1007/978-3-319-08867-9\_50}}.

\end{thebibliography}

\end{document}